%% file: main.tex
\documentclass{article}
\input{packages}
\input{macros}
\title{Invasion Dynamics in the Biased Voter Process}

\author{
Loke Durocher$^1$
\and
Panagiotis Karras$^1$
\and
Andreas Pavlogiannis$^1$\And
Josef Tkadlec$^2$\\
\affiliations
$^1$Aarhus University, Aabogade 34, Aarhus, Denmark\\
$^2$Harvard University, 1 Oxford Street, Cambridge, USA\\
\emails
\{panos,pavlogiannis\}@cs.au.dk,
tkadlec@math.harvard.edu
}
\renewcommand{\smallskip}{}

\begin{document}
\pagestyle{plain}
\maketitle
\input{sections/00.abstract}
\input{sections/01.intro}
\input{sections/02.prelims}
\input{sections/03.comp-fpras}
\input{sections/04.0.bias_positive}
\input{figures/fig_experiments}
\input{sections/05.bias_negative}
\input{sections/06.experiments}
\input{sections/07.conclusion}

%\section*{Acknowledgements}
%\clearpage
%\AtBeginEnvironment{thebibliography}{\linespread{0.99}\selectfont}
\bibliographystyle{named}
\bibliography{bibliography}

\newpage
\appendix
\input{sections/appendix}
\end{document}

%% file: packages.tex
\usepackage{ijcai22}
\usepackage{times}
\usepackage{soul}
\usepackage{url}
\usepackage[hidelinks]{hyperref}
\usepackage[utf8]{inputenc}
\usepackage[small]{caption}
\usepackage{graphicx}
\usepackage{amsmath}
\usepackage{amsthm}
\usepackage{booktabs}

\usepackage{xfrac} % for \sfrac
\usepackage{amsmath, amsthm,amssymb,amsfonts}
\usepackage{thmtools}
\usepackage{thm-restate}
\usepackage{parskip}
\usepackage{enumitem}
\usepackage[ruled, linesnumbered]{algorithm2e}

\usepackage[capitalise]{cleveref}

\usepackage{tikz}
\usetikzlibrary{positioning,arrows,automata,shapes,decorations,decorations.markings,calc, matrix,decorations.pathmorphing, patterns,backgrounds}

%% file: macros.tex
\def\eps{\varepsilon}

\def\bigO{\mathcal{O}}

\def\R{\mathbb{R}}

\newlist{compactenum}{enumerate}{3} % 3 is max-depth
\setlist[compactenum]{label=\arabic*., nosep,leftmargin=*}
\crefname{compactenumi}{Item}{Items}

\DeclareMathOperator*{\argmax}{arg\,max}
\DeclareMathOperator*{\argmin}{arg\,min}

\newtheorem{lemma}{Lemma}
\newtheorem{corollary}{Corollary}

\makeatletter
\newenvironment{myproof}[1][\proofname]{\par
  \pushQED{\qed}%
  \normalfont \partopsep=\z@skip \topsep=\z@skip
  \trivlist
  \item[\hskip\labelsep
        \itshape
    #1\@addpunct{.}]\ignorespaces
}{%
  \popQED\endtrivlist\@endpefalse
}
\makeatother

\newcommand{\fit}{f}

\newcommand{\Reals}{\mathbb{R}}
\newcommand{\Neighbors}{N}
\newcommand{\Configuration}{X}
\newcommand{\RandomConfiguration}{\mathcal{X}}
\newcommand{\RandomYConfiguration}{\mathcal{Y}}
\newcommand{\RandomZConfiguration}{\mathcal{Z}}
\newcommand{\SeedSet}{S}

\newcommand{\fp}{\operatorname{fp}}
\newcommand{\ft}{\operatorname{T}}
\newcommand{\E}{\mathbb{E}}

\newcommand{\Temp}{\mathcal{T}}

\newcommand{\Paragraph}[1]{{\smallskip\noindent\bf #1}}

\newcommand{\PClass}{\ensuremath{\operatorname{\mathbf{P}}}}
\newcommand{\NP}{\ensuremath{\operatorname{\mathbf{NP}}}}

\newcommand{\Weight}{w}

\newcommand\numberthis{\addtocounter{equation}{1}\tag{\theequation}}

\newcommand{\cGX}{c(G,\SeedSet)}
\newcommand{\puv}{p_{u\to v}}
\newcommand{\pvu}{p_{v\to u}}
\newcommand{\contrib}{\delta}

%% file: sections/00.abstract.tex
\begin{abstract}
The \emph{voter process} is a classic stochastic process that models the invasion of a mutant trait~$A$ (e.g., a new opinion, belief, legend, genetic mutation, magnetic spin) in a population of agents (e.g., people, genes, particles) who share a resident trait~$B$, spread over the nodes of a graph. An agent may adopt the trait of one of its neighbors at any time, while the \emph{invasion bias}~$r\in(0,\infty)$ quantifies the stochastic preference towards ($r>1$) or against ($r<1$) adopting~$A$ over~$B$. Success is measured in terms of the \emph{fixation probability}, i.e., the probability that eventually all agents have adopted the mutant trait~$A$. In this paper we study the problem of \emph{fixation probability maximization} under this model: given a budget~$k$, find a set of~$k$ agents to initiate the invasion that maximizes the fixation probability. We show that the problem is \NP-hard for both~$r>1$ and~$r<1$, while the latter case is also inapproximable within any multiplicative factor. On the positive side, we show that when~$r>1$, the optimization function is submodular and thus can be greedily approximated within a factor~$1-1/e$. An experimental evaluation of some proposed heuristics corroborates our results.
\end{abstract}

%% file: sections/01.intro.tex
\section{Introduction}\label{sec:intro}

Invasion processes are a standard framework for modeling the spread of novel traits in populations. The population structure is represented as a graph, with nodes occupied by agents, while the edge relation captures some sort of agent connection, such as spatial proximity, friendship, direct communication, or particle interaction. The invasion dynamics define the rules of trait diffusion from an agent to its neighbors. These dynamics raise a natural optimization problem that calls to choose a set of at most~$k$ initial trait-holders that maximizes a goal involving the eventual spread of the trait in the population. This type of problem has been studied extensively with a focus on influence spread on social network dynamics~\cite{Domingos2001,Kempe2003,Mossel2007}; see~\cite{IM-Survey} for a survey.

The \emph{voter process} is a fundamental and natural probabilistic diffusion model that arises in evolutionary population dynamics, where the diffused trait is a mutation and the initial trait-holders are mutants. Also known as imitation updating in evolutionary game theory~\cite{Ohtsuki2006}, it was initially introduced to study particle interactions~\cite{Liggett1985} and territorial conflict~\cite{Clifford1973}; since then, it has found numerous applications in the spread of genetic mutations (the death-birth Moran process)~\cite{Hindersin2015,Tkadlec2020}, social dynamics~\cite{Castellano2009}, robot coordination and swarm intelligence~\cite{Talamali2021}. In the long run, the process leads to a \emph{consensus state}~\cite{EvenDar2007}, in which the mutant trait is either fixated (i.e., adopted by all agents), or gone extinct. 
The calibre of the invasion is quantified in terms of its \emph{fixation probability}~\cite{Kaveh2015,brendborg2022fixation}; thus, the goal of optimization in this context is to maximize the probability that the mutant trait spreads in the whole population. Unlike other models of spread, the voter process accounts for \emph{active} resistance to the invasion:~an agent carrying the mutant trait may not only forward it to its neighbors, but also lose it due to exposure to the resident trait from its neighbors. This model exemplifies settings such as individuals that may switch opinions, magnetic spins that may undergo transient states, and sub-species that compete in a gene pool.

In this paper, we study the problem of \emph{fixation probability maximization} in the context of the \emph{biased voter process}:~find a set of~$k$ agents initially having a \emph{mutant trait}~$A$ that maximizes the fixation probability of~$A$ in a population with a preexisting \emph{resident trait}~$B$. The \emph{invasion bias}~$r\in(0,\infty)$ quantifies the stochastic preference of each agent towards~($r>1$) or against~($r<1$) adopting~$A$ over~$B$. Despite extensive interest in this type of optimization problem and the appeal and generality of the voter process, to our knowledge, this problem in the context of the voter process has only been studied in the restricted case~$r=1$, which admits a simple polynomial-time algorithm~\cite{EvenDar2007}. Yet biased invasions naturally occur in most settings:~for genetic mutations, $r$~captures the fitness (dis)advantage conferred by the new mutation~\cite{Allen2020b}; for magnetic spins, $r$~reveals the presence of a magnetic field~\cite{Antal2006}; for social traits, $r$~accounts for the societal predisposition towards a new idea~\cite{Bhat2019}.
% pk 07jan2022

% contribution
Our contributions can be summarized as follows:

\begin{compactenum}
\item We show that computing the fixation probability on undirected graphs admits a fully polynomial-time approximation scheme (FPRAS).
\item We show that, notwithstanding the neutral case where~$r=1$, fixation maximization is $\NP$-hard both for~$r>1$ and for~$r<1$, while the latter case is also inapproximable within any multiplicative factor.
\item On the positive side, we show that, when~$r>1$, the optimization function is submodular and thus can be greedily approximated within a factor~$1-\sfrac{1}{e}$.
\end{compactenum}

%% file: sections/02.prelims.tex
\section{Preliminaries}\label{sec:preliminaries}

\paragraph{Structured populations.} Following the literature standard, we represent a population structure by a (generally, directed and weighted) graph~$G=(V,E, \Weight)$ of~$n=|V|$ nodes. 
Each node stands for a location in some space, while a population of~$n$ agents is spread on~$G$ with one agent per node. 
Given a node~$u$, we denote by~$\Neighbors(u)=\{v\in V\colon (v,u)\in E \}$ the \emph{neighbors} of~$u$, and by~$\deg(u)=|\Neighbors(u)|$ its \emph{degree}. 
The weight function~$\Weight\colon E\to \Reals_{\geq 0}$ maps every edge to a non-negative number. 
As a convention, we write~$\Weight(v,u)=0$ if~$(v,u)\not \in E$. 
We consider graphs for which the sub-graph induced by the support of~$\Weight$ is strongly connected --- which is necessary for the voter process to be well-defined. 
We say that~$G$ is \emph{undirected} if~$E$ is symmetric and irreflexive, and~$\Weight(v,u)=1$ for all~$(v,u)\in E$. 
For simplicity, when~$G$ is undirected we denote it by~$G=(V,E)$.

\paragraph{The voter invasion process.} The classic voter process models the invasion of a mutant \emph{trait}~$A$ into a homogeneous population~\cite{Clifford1973,Liggett1985}. The process distinguishes between two types of agents, \emph{residents}, who carry a resident trait~$B$, and \emph{mutants}, who carry the mutant trait~$A$. A \emph{seed set}~$\SeedSet\subseteq V$ defines the agents that are initially mutants. The \emph{invasion bias}~$r>0$ is a real-valued parameter that defines the relative competitive advantage of mutants in this invasion. When~$r>1$ (resp.,~$r<1$), the mutants have a competitive advantage (resp., competitive disadvantage), and the process is biased towards (resp., against) invasion, while~$r=1$ yields the neutral case.
A \emph{configuration}~$\Configuration\subseteq V$ defines the set of mutants at a given time.
The \emph{fitness} of the agent occupying a node~$u \in V$ is defined as
\[
\fit_{\Configuration}(u)=
\begin{cases}
r, & \text{ if }u\in \Configuration\\
1, & \text{ otherwise.}
\end{cases}
\]

The discrete-time \emph{voter invasion process}~\cite{Antal2006,EvenDar2007}, a.k.a. \emph{Moran death-birth process}~\cite{Tkadlec2020,Allen2020b}, with seed set~$\SeedSet$ is a stochastic process $(\RandomConfiguration_i)_{i\geq 0}$, where~$\RandomConfiguration_i\subseteq V$ is a random \emph{configuration}. 
Initially, $\RandomConfiguration_0=\SeedSet$, and,
given the current configuration~$\Configuration$, we obtain
$\RandomConfiguration_{i+1}$ in two stochastic steps.
\begin{compactenum}
\item \emph{(Death event):} a focal agent on some node~$u$ is chosen with uniform probability~$\sfrac{1}{n}$.
\item \emph{(Birth event):} the agent on~$u$ adopts the trait of one of its neighbors, on $v$, with probability:
\[
\frac{\fit_{\Configuration}(v)\cdot \Weight(v,u)}{\sum_{x\in \Neighbors(u)}\fit_{\Configuration}(x)\cdot \Weight(x,u)} .
\]
\end{compactenum}
Note that the focal agent may assume its pre-existing trait, and generally, the mutant set may grow or shrink in each step.
We often identify agents with the nodes they occupy and use expressions such as ``node~$v$ reproduces onto~$u$'' to indicate that the agent occupying node~$u$ adopts the trait of the agent at node~$v$. See~\cref{fig:voter} for an illustration.
% pk 08jan2022

\input{figures/fig_voter}

\paragraph{Fixation probability.} As~$G$ is strongly connected, one type will spread to all nodes in finitely many steps with probability~1. The mutant invasion succeeds if the mutant trait spreads, an event called \emph{fixation}. Given a graph~$G$, a bias~$r$, and an initial seed set~$\SeedSet$, the \emph{fixation probability}, denoted as~$\fp^G_r(\SeedSet)$, is the probability that the mutants starting from~$\SeedSet$ eventually fixate on~$G$. See \cref{fig:twostar} for an illustration.

\input{figures/fig_twostar}

Whereas the computational complexity of determining~$\fp^G_r(\SeedSet)$ is unknown, the function is approximately computable by a simulation of the invasion process. In the next section we will show that when~$G$ is undirected, such a simulation yields an efficient approximation scheme.
% pk 08jan2022; END of Page 2.

% no need to say this is something "standard".
\paragraph{Fixation maximization.} We formulate and study the optimization problem of selecting a seed set that leads to a most successful invasion, i.e., maximizes the fixation probability, subject to cardinality constraints. 
Formally, given a graph~$G$, a bias~$r$ and a budget~$k$, the task is to compute the set
\[
\SeedSet^*=\argmax_{\SeedSet\subseteq V, |\SeedSet|= k} \fp^G_{r}(\SeedSet)
\numberthis\label{eq:optimization}
\]
The corresponding decision question is: Given a budget~$k$ and a threshold~$\alpha\in[0,1]$ determine whether there exists a seed set~$\SeedSet$ with~$|\SeedSet|=k$ such that~$\fp^G_{r}(\SeedSet) \geq \alpha$. As~\cref{fig:twostar} illustrates, different seed sets yield considerably different fixation probabilities, while the optimal choice varies depending on~$r$, hence the intricacy of the problem.

%% file: figures/fig_voter.tex
\begin{figure}[!ht]
\includegraphics[width=0.95\linewidth]{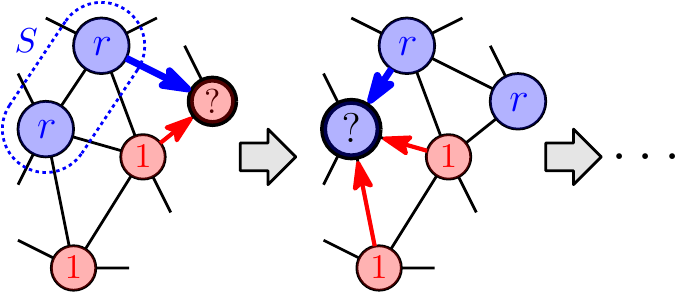}
\caption{Initially, each node is occupied by a resident (red), except for a seed set~$S$ of nodes occupied by mutants (blue). In each step of the voter process with bias~$r$, one random node adopts the type of one of its neighbors, with mutants having a relative (dis)advantage~$r$.
}\label{fig:voter}
\end{figure}

%% file: figures/fig_twostar.tex
\begin{figure}[!ht]
\includegraphics[width=0.95\linewidth]{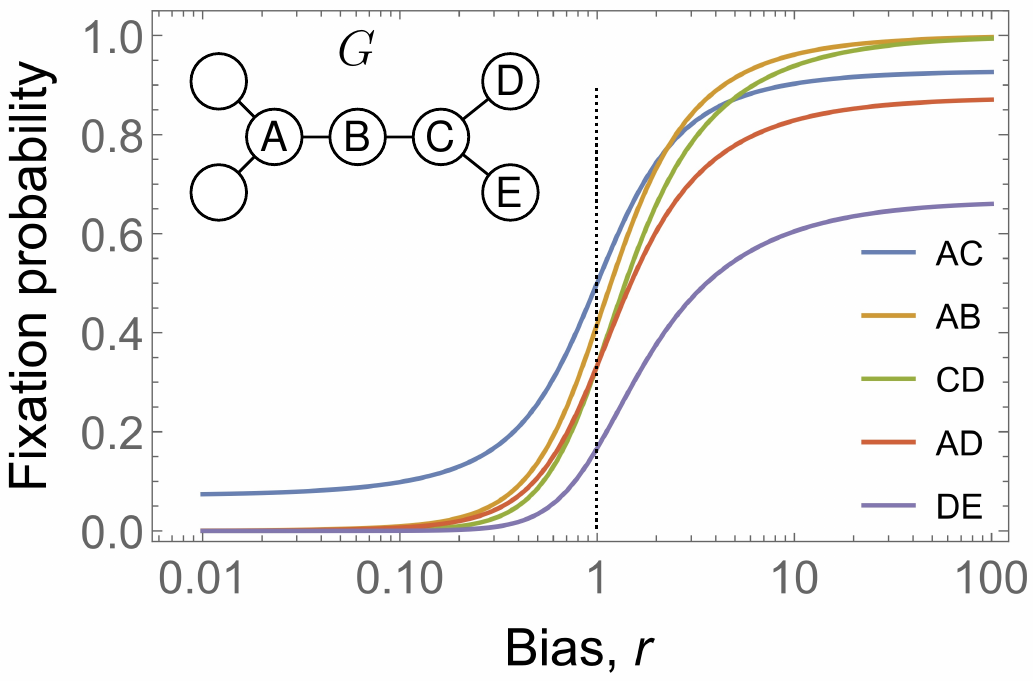}
\caption{A graph~$G=(V,E)$ of~7 nodes and fixation probability~$\fp^G_r(\SeedSet)$ for~5 different subsets~$\SeedSet$ of~2 nodes vs. bias~$r\in[0.01,100]$. We have $\fp^G_{r\to 0}(\SeedSet)\to 0$, unless $\SeedSet$ forms a vertex cover (\cref{lem:negative_bias_vc}) and $\fp^G_{r\to \infty}(\SeedSet)\to 1$ iff $\SeedSet$ contains two \mbox{adjacent}~nodes.
}\label{fig:twostar}
\end{figure}

%% file: sections/03.comp-fpras.tex
\section{Computing the Fixation Probability}\label{sec:fpras}

For~$r=1$, the fixation probability is additive~\cite{Broom2010}. Thus, given a graph~$G=(V,E,\Weight)$ and a budget~$k$, the set~$\SeedSet^*$ that maximizes~$\fp^G_{r=1}(\SeedSet)$ is obtained by selecting the~$k$ largest values from the vector~$\langle \fp^G_{r=1}(\{v\}) \mid v\in V \rangle$, determined by solving a system of~$n$ linear equations~\cite{allen2015molecular}; thus the problem can be solved in polynomial time. 
When~$G$ is undirected, that system admits a closed-form solution~\cite{Broom2010,Maciejewski2014}:
%When~$G$ is undirected, that system admits a closed-form solution~\cite{Broom2010}:
$\fp^G_{r=1}(\{v\}) = \sfrac{\deg(v)}{2|E|}$.

%solution~\cite{Broom2010,Maciejewski2014}:}
%\[
%\fp^G_{r=1}(\{v\}) = \frac{\deg(v)}{2|E|}.
%\]

For~$r\ne 1$, the complexity of determining~$\fp^G_r(\SeedSet)$ is open. Here we show that, for undirected graphs, it can be efficiently approximated by a fully-polynomial randomized approximation scheme (FPRAS). 
The key component of our result is a polynomial upper bound on the expected number of steps~$\ft^G_r(\SeedSet)$ until one of the types fixates;
%(see~\cref{thm:ub-time}); 
an analogous approach has been used for the so-called Moran Birth-death process~\cite{Diaz2014}.
% For that process, subsequent work improved the bound on the time as well as the implementation of the simulation~\cite{ann2020phase}.

\begin{restatable}{theorem}{thmubtime}\label{thm:ub-time}
Let~$G=(V,E)$ be an undirected graph with~$|V|=n$ nodes, $|E|=m$ edges, and maximum degree~$\Delta$. Let~$\SeedSet\subseteq V$ be the initial seed set and~$r\ne 1$. Then~$\ft^G_r(\SeedSet)=\bigO(\Delta nm) =\bigO(n^4)$.
\end{restatable}
\begin{myproof}%[Proof sketch]
%This is a proof sketch, see~\cref{sec:appendix} for details.
Without loss of generality, suppose~$r>1$ (otherwise we swap the role of mutants and residents).
Given a current configuration~$X$, we define a potential~$\phi(X)\colon 2^V\to \R$ as~$\phi(X)=\sum_{v\in X} \deg(v)$. 
We show that, as we run the process, $\phi(X)$ increases in a controlled manner.
To that end, we distinguish two types of configurations: we say that a configuration is \textit{bad} if every edge~$uv\in E$ connects nodes of different types, and \textit{good} otherwise. Note that bad configurations exist iff~$G$ is bipartite, in which case there are exactly two of them. In~\cref{sec:appendix}, we show that:
\begin{compactenum}
\item When at a bad configuration, the next configuration is good and the potential does not change in expectation.
\item When at a good configuration, in one step the potential increases by at least~$s=(r-1)/(r\Delta n)$ in expectation.
\end{compactenum}
All in all, this implies that, in expectation, the potential increases by at least~$s$ over any two consecutive steps (until fixation occurs). Since at all times the potential is non-negative and upper-bounded by~$U=\sum_{v\in V}\deg(v)=2m$, by a standard drift theorem for submartingales~\cite{He2001,lengler2020drift} the expected number of steps is at most:
\[ 
\ft^G_r(\SeedSet) \le 2\frac{U}{s} = \frac{4r}{r-1} \Delta nm.\qedhere
\]
\end{myproof}

%\cref{thm:ub-time} yields the following corollary, based on simulating the process a sufficient number of times and reporting the empirical average.

\cref{thm:ub-time} yields the following corollary based on the Monte Carlo method.

%\Andreas{May we remove $r\neq 1$? The corollary is still correct and simpler}

\begin{corollary}\label{thm:fpras}
When~$G=(V,E)$ is undirected, $\SeedSet\subseteq V$, the function~$\fp^G_r(\SeedSet)$ admits a FPRAS for any~$r$.
\end{corollary}
% pk 09jan2022

%% file: sections/04.0.bias_positive.tex
\section{Invasion with Favorable Bias}\label{sec:positive_bias}

In this section we study the voter process for~$r>1$. We show that the problem is \NP-hard (\cref{subsec:opt-hardness}) but the fixation probability is monotone and submodular and can thus be efficiently approximated (\cref{subsec:submodularity}).

\input{sections/04.1.opt-hardness}
\input{sections/04.2.submodularity}

%% file: sections/04.1.opt-hardness.tex
\subsection{Hardness of Optimization}\label{subsec:opt-hardness}

Here we consider the weakly biased process, i.e., the case~$r = 1 + \eps$ for small~$\eps > 0$. Since~$\fp^G_{r}(\SeedSet)$ is a continuous and smooth function of~$r$, its Taylor expansion around~$r = 1$ is:
\[
\fp^G_{r=1+\eps}(\SeedSet) = \fp^G_{r=1}(\SeedSet) + \eps\cdot \cGX +\bigO(\eps^2),
\numberthis\label{eq:taylor}
\]
where~$\cGX$ is a constant independent of~$\eps$. Recall that %, for~$r=1$,
$\fp^G_{r=1}(\SeedSet) = \sum_{v\in \SeedSet} \sfrac{\deg(v)}{2|E|}$%~\cite{Broom2010}.
~\cite{Broom2010,Maciejewski2014}.
For~$d$-regular graphs, this yields~$\fp^G_{r=1} (\SeedSet) = \sfrac{d|S|}{2|E|}$, hence the first term in~\cref{eq:taylor} is constant across all equal-sized seed sets~$\SeedSet$. We thus have the following lemma.

\begin{restatable}{lemma}{lemmaxder}\label{lem:max_der}
For any undirected regular graph~$G$, there exists a small constant~$\eps>0$ such that:
\[
\argmax_{S\subseteq V, |S|= k} \fp^G_{r=1+\eps}(\SeedSet)  \subseteq \argmax_{S\subseteq V, |S|= k} \cGX .
\]
\end{restatable}

In light of~\cref{lem:max_der}, we establish the hardness of fixation maximization by computing~$\cGX$ analytically on regular graphs, and showing that maximizing it is \NP-hard. 
McAvoy and Allen~\shortcite{Mcavoy2021} recently showed that, given~$G$ and~$\SeedSet$, the constant~$\cGX$ can be defined through a system of~$n^2$ linear equations. Here we determine~$\cGX$ explicitly for regular graphs.
%First, we note that the system defining $\cGX$ takes a simpler form for $d$-regular graphs.
For each node~$v\in V$, let~$x_v=1$ iff $v\in \SeedSet$. First, we restate~\cite[Proposition~2]{Mcavoy2021} in the special case of~$d$-regular graphs and adapt it to our notation.

\begin{lemma}\label{lem:mcavoy}
Fix an undirected~$d$-regular graph~$G=(V,E)$ on~$n$ nodes and an initial configuration~$\SeedSet\subseteq V$. Then
\[
\cGX=\frac1{2d^2n^2} \sum_{u,v\in V} |\Neighbors(u)\cap \Neighbors(v)|\cdot  \eta_{uv},
\numberthis\label{eq:mcavoy_cgx}
\]
where~$\eta_{uv}$ comprise the unique solution to the linear system
\[
\eta_{uv}=\begin{cases}
%\frac n2(x_u+x_v-2x_ux_v)  %\displaystyle
\frac n2|x_u-x_v|           + \frac1{2d} \sum_{w\in \Neighbors(u)} \eta_{wv} &\\
\phantom{\frac n2|x_u-x_v|} + \frac1{2d} \sum_{w\in \Neighbors(v)} \eta_{uw} &\text{if $u\ne v$,}\\
0 &\text{if $u=v$.}
\numberthis\label{eq:mcavoy_system}
\end{cases}
\]
\end{lemma}
Intuitively, each~$\eta_{uv}$ expresses the expected number of steps in which nodes~$u$ and~$v$ are occupied by a mutant and resident, respectively. We say that an edge~$e\in E$ is \textit{active} if its endpoints are occupied by heterotypic agents. The following lemma establishes that, on regular graphs, $\cGX$ is decreasing in the number of active edges induced by~$\SeedSet$.

\begin{restatable}{lemma}{lemfpweakregular}\label{lem:fp-weak-regular}
Let~$G=(V,E)$ be a~$d$-regular graph on~$n$ nodes and~$\SeedSet\subseteq V$ an initial configuration of size~$|\SeedSet|=k$. Let~$a(\SeedSet)$ be the number of \emph{active} edges induced by~$\SeedSet$. Then
\[
\cGX = \frac{k(n-k)}{2n} - \frac{a(\SeedSet)}{2dn}.
\]
\end{restatable}
\begin{myproof}
Let~$|uv|$ be the shortest-path distance between~$u$ and~$v$ in~$G$. We proceed in two steps. First, we claim that summing all the~$n^2$ equations in the system of~\cref{eq:mcavoy_system}, we obtain

\[\sum_{u,v\in V}\eta_{uv}=\frac n2\cdot 2k(n-k) + \sum_{u,v\in V}\eta_{uv} - \frac1{2d}\sum_{|uv|=1}2\eta_{uv}.\numberthis\label{eq:sum-all1}\]

Indeed:
\begin{compactenum}
\item We have $\sum_{u,v\in V}|x_u-x_v|=2k(n-k)$, since each ordered pair of
heterotypic agents
%an agent~$\in \SeedSet$ and one~$\notin\SeedSet$
is counted once.
\item For a fixed pair~$(u,v)$ with~$|uv|\ge 2$, term $\frac1{2d}\eta_{uv}$ appears on the right-hand side of~$2d$ equations, namely~$d$ equations with~$\eta_{wv}$ on the left-hand side and~$u \!\in\! N(w)$ and~$d$ equations with~$\eta_{uw}$ on the left-hand side and~$v \!\in\! N(w)$.
% pk fixed this.
\item For a fixed pair~$(u,v)$ with~$|uv|=1$, term $\frac1{2d}\eta_{uv}$ appears on the right-hand side of~$2d-2$ equations, namely the above~$2d$ equations, except for 2 cases where~$w\in\{u,v\}$.
\end{compactenum}
Rearranging~\cref{eq:sum-all1}, we obtain
\[\sum_{|uv|=1}\eta_{uv} = d n k (n-k)\numberthis\label{eq:sum-all2}.\]

Second, we claim that summing the~$2|E|$ equations in the system of~\cref{eq:mcavoy_system} where~$|uv|=1$, we obtain
%\Andreas{@Pepa:Can you give more details in the derivation}
\[\sum_{|uv|=1}\eta_{uv} = \frac n2\cdot 2a(\SeedSet) + \frac1{2d}\sum_{u,v\in V}2|\Neighbors(u)\cap \Neighbors(v)|\cdot \eta_{uv},
\]
ergo, using~\cref{eq:sum-all2}
\[\sum_{u,v\in V}|\Neighbors(u)\cap \Neighbors(v)|\cdot \eta_{uv}
= d^2nk(n-k) - dn\cdot a(\SeedSet).
\numberthis\label{eq:sum-dist1}
\]
Indeed,
\begin{compactenum}
\item We have~$\sum_{|uv|=1}|x_u-x_v|=2a(\SeedSet)$, since each active edge is counted exactly twice, once in each direction.
\item For a fixed pair~$(u,v)$ with~$|uv|=1$, let~$N(u,v) = N(u)\cap N(v)$ be the set of common neighbors of~$u$ and~$v$. The term~$\frac1{2d}\eta_{uv}$ then appears on the right-hand side of precisely~$2|N(u,v)|$ equations, namely~$|N(u,v)|$ equations with~$\eta_{wv}$ on the left-hand side and~$|N(u,v)|$ equations with~$\eta_{uw}$ on the left-hand side, where~$w\in N(u,v)$.
\end{compactenum}
From~\cref{lem:mcavoy} and~\cref{eq:sum-dist1}, we obtain the desired
\begin{align*}
    \cGX &=\frac1{2d^2n^2} \sum_{u,v\in V}|\Neighbors(u)\cap \Neighbors(v)|\cdot \eta_{uv}\\
    &=\frac{d^2nk(n-k)-a(\SeedSet)dn}{2d^2n^2} = \frac{k(n-k)}{2n} - \frac{a(\SeedSet)}{2dn}.
\end{align*}
\end{myproof}

Thus, given a budget~$k$, we maximize the fixation probability at the limit~$r\to 1$, by minimizing the number of active edges. We are now ready to establish the main result in this section.

\begin{restatable}{theorem}{thmpositivebiasnphard}\label{thm:positive_bias_np_hard}
Fixation maximization in the biased voter process with~$r>1$ is~$\NP$-hard, even on regular undirected graphs.
\end{restatable}
\begin{myproof}
\cref{lem:max_der} and \cref{lem:fp-weak-regular} imply that, for any regular undirected graph~$G$, there exists a small~$\eps>0$ for which
\[
\argmax_{\SeedSet\subseteq V, |\SeedSet|= k} \fp^G_{r=1+\eps}(\SeedSet) \subseteq \argmin_{\SeedSet\subseteq V, |\SeedSet|= k}a(\SeedSet)
\]
i.e., the fixation probability is maximized by a seed set that minimizes the number of active edges. For~$k=\sfrac{n}{2}$, this is known as the minimum bisection problem, which is $\NP$-hard even on 3-regular graphs~\cite{bui1987graph}.
\end{myproof}

%% file: sections/04.2.submodularity.tex
\subsection{Monotonicity and Submodularity}\label{subsec:submodularity}

A real-valued set function~$f$ is called monotone if for any two sets~$A$ and~$B$ with~$A\subseteq B$, we have~$f(A)\leq f(B)$. Further, $f$ is called submodular if for all sets~$A$ and~$B$ we have
\[
f(A)+f(B)\geq f(A\cup B) + f(A\cap B) .
\numberthis\label{eq:submodularity}
\]
Monotone submodular functions, even if hard to maximize, admit efficient approximations. Here we complement the hardness of~\cref{thm:positive_bias_np_hard} 
by showing that~$\fp^G_{r \geq 1} (\SeedSet)$ is monotone and submodular. 
Although submodularity is well-known for some \emph{progressive} invasion processes~\cite{Kempe2003}, recall that our setting is \emph{non-progressive}, i.e., the mutant set may grow or shrink in any step. 
This non-progressiveness renders submodularity a non-trivial property.

\Paragraph{Node duplication.} 
Towards establishing monotonicity and submodularity, it is convenient to view the process in a slightly different but equivalent way, via \emph{node duplication}. Consider the voter process~$\mathcal{M}$ on a graph~$G$ with bias~$r>1$. The node duplication view of~$\mathcal{M}$ is another stochastic process~$\mathcal{M}'=(\RandomConfiguration'_i)_i$ on~$G$, defined as follows. Let~$\Configuration\subseteq V$ be the current configuration.
\begin{compactenum}
\item A node~$u$ is chosen for death uniformly at random (i.e., this step is identical to~$\mathcal{M}$).
\item We modify~$G=(V,E,\Weight)$ to~$G'=(V',E', \Weight')$ as follows. For every node~$v \in \Neighbors(u) \cap \Configuration$, we create a duplicate~$v'$ of~$v$, and insert an edge~$(v',u)$ in~$G'$ with weight~$\Weight'(v',u) = \Weight(v,u)$. We associate~$v'$ with fitness~$\fit'_{\Configuration} (v') = r-1$, while every other node~$v\in V$ gets fitness $\fit'_{\Configuration}(v)=1$. % pk corrected from r-1 to 1.
Finally, we execute a stochastic birth step as in the standard voter process, i.e., we choose a node~$v$ to propagate to~$u$ with probability
\[
\frac{\fit'_{\Configuration}(v)\cdot \Weight'(v,u)}{\sum_{x}\fit'_{\Configuration}(x)\cdot \Weight'(x,u)}
\]
The new configuration~$\RandomConfiguration'_{i+1}$ contains~$u$ if either some mutant node~$v$ or its duplicate~$v'$ was chosen for reproduction. It is straightforward to see that this modified birth step preserves the probability distribution of random configurations, and thus~$\mathcal{M}$ and~$\mathcal{M}'$ are equivalent.
\end{compactenum}

\paragraph{Monotonicity.} Using the equivalence between the voter process and its variant with node duplication, we now establish the monotonicity of the fixation probability.

\begin{restatable}{lemma}{lemmonotonicity}\label{lem:monotonicity}
For any graph~$G$ and bias~$r$, for any two seed sets~$\SeedSet_1$ and~$\SeedSet_2$ with~$\SeedSet_1 \subseteq \SeedSet_2$, we have~$\fp^G_{r}(\SeedSet_1) \leq \fp^G_{r}(\SeedSet_2)$.
\end{restatable}
\begin{myproof}
First assume~$r>1$. Consider two processes~$\mathcal{M}_1=(\RandomConfiguration^1_i)_i$ and~$\mathcal{M}_2 = (\RandomConfiguration^2_i)_i$ starting from~$\RandomConfiguration^1_0 = \SeedSet_1$ and~$\RandomConfiguration^2_0 = \SeedSet_2$, respectively. We establish a coupling between~$\mathcal{M}_1$ and~$\mathcal{M}_2$ that satisfies~$\RandomConfiguration^1_i \subseteq  \RandomConfiguration^2_i$ for all~$i\geq 0$, from which the lemma follows.

In each step, we first choose the same node~$u$ for death in $\mathcal{M}_2$ and~$\mathcal{M}_1$ with uniform probability. 
Then, we execute a birth step in~$\mathcal{M}_2$. 
If the reproducing node~$v$ is present in~$\mathcal{M}_1$, we perform the same update in~$\mathcal{M}_1$. 
Otherwise, if~$v$ is a mutant duplicate absent from~$\mathcal{M}_1$, we perform an independent birth step in~$\mathcal{M}_1$. 
This coupling maintains the invariant~$\RandomConfiguration^1_i \subseteq \RandomConfiguration^2_i$, as whenever~$\mathcal{M}_1$ takes an independent step, a mutant has reproduced in~$\mathcal{M}_2$. Moreover, the coupling is indeed transparent to~$\mathcal{M}_1$ and~$\mathcal{M}_2$, as every birth event occurs with probability proportional to its agent's fitness.

The monotonicity for all~$r$ follows by symmetry, as when~$r<1$ we can exchange the roles of mutants and residents. Using the above argument, it follows that the fixation probability of the residents grows monotonically in~$V\setminus \SeedSet$ and, in reverse, decreases monotonically as~$V \setminus \SeedSet$ shrinks, hence the fixation probability of the mutants grows monotonically in~$\SeedSet$.
\end{myproof}

\Paragraph{Submodularity.} 
We now show submodularity, i.e., we have
\[
\fp^G_{r}(A)+\fp^G_{r}(B)\geq \fp^G_{r}(A\cup B) + \fp^G_{r}(A\cap B) .
\numberthis\label{eq:submodularity_fp}
\]
To this end, it is convenient to consider the following more refined view of the invasion process. Given an initial seed set~$\SeedSet=A\cup B$, we keep track not only of the mutant set, but also the subset of mutants that are copies of 
those agents initially in $A$, and those initially in $B$.
Consider this refined view, and execute the invasion beyond the point where mutants fixate. 
With probability 1, every execution that results in the fixation of $S$ eventually leads to the fixation of $A$ or $B$ (possibly both, assuming that $A\cap B\neq \emptyset$).
We can thus compute the fixation probability of $\SeedSet$ by summing over each execution that ends in the fixation of at least one of $A$, $B$.
This ``extended'' process is instrumental in proving submodularity.

\begin{restatable}{lemma}{lemsubmodular}\label{lem:submodular}
For any graph~$G$ and bias~$r\geq 1$, the fixation probability~$\fp^G_{r}(\SeedSet)$ is submodular.
\end{restatable}
\begin{myproof}
Consider any two sets~$A,B\subseteq V$. 
Let~$\mathcal{M}_1$, $\mathcal{M}_2$, $\mathcal{M}_3$ and~$\mathcal{M}_4$ be the four (extended) invasion processes, under node duplication, that correspond to initial seeds~$A$, $B$, $A\cup B$ and~$A\cap B$, respectively. For~$j\in [4]$, let~$\RandomConfiguration^j_i$ be the~$i$-th random configuration of~$\mathcal{M}_j$. Moreover, let~$\RandomYConfiguration_i$ (resp., $\RandomZConfiguration_i$) be the set of nodes of $\mathcal{M}_3$ in step~$i$  that are occupied by mutant agents  who are descendants of some agent initially in~$A$ (resp., $B$). We establish a four-way coupling between all~$\mathcal{M}_j$ that preserves the following invariants:
\[
\text{(i)}~\RandomYConfiguration_i\subseteq \RandomConfiguration^1_i
\qquad
\text{(ii)}~\RandomZConfiguration_i\subseteq \RandomConfiguration^2_i
\quad
\text{and}
\quad
\text{(iii)}~\RandomConfiguration^4_i\subseteq \RandomConfiguration^1_i\cap \RandomConfiguration^2_i .
\]
Given our extended invasion process, any execution that leads to the fixation of $A\cup B$ in $\mathcal{M}_3$ will eventually witness the fixation of at least one of $A$ or $B$.
Invariants (i) and (ii), then,  guarantee that any fixating execution in~$\mathcal{M}_3$ is also fixating in at least one of~$\mathcal{M}_1$ and~$\mathcal{M}_2$.
Invariant (iii) guarantees that any fixating run execution in~$\mathcal{M}_4$ is fixating in both~$\mathcal{M}_1$ and~$\mathcal{M}_2$. 
Hence, the three invariants imply \cref{eq:submodularity_fp}, i.e., submodularity.

We now establish the coupling. In each step~$i$, we choose the same node~$u$ in all~$\mathcal{M}_j$ to die uniformly at random. Then, as long as~$\RandomYConfiguration_i,\RandomZConfiguration_i \neq \emptyset$ we perform the following steps.
\begin{compactenum}
\item Choose a node~$v$ to reproduce in~$\mathcal{M}_3$ using node duplication (hence~$v$ is either an original or a duplicate node). For each process~$\mathcal{M}_j$ in which~$v$ is a node (either original or duplicate), propagate the agent from~$v$ to~$u$. Note that this step updates~$\mathcal{M}_3$ and at least one of~$\mathcal{M}_1$ and~$\mathcal{M}_2$, while if it updates~$\mathcal{M}_4$, it also updates both~$\mathcal{M}_1$ and~$\mathcal{M}_2$.
\item If~$v$ is not a node in one of~$\mathcal{M}_1$ or~$\mathcal{M}_2$, then perform an independent birth step in that process, choosing some node~$x$ (either original or duplicate). If~$x$ is present in~$\mathcal{M}_4$, apply the birth event in~$\mathcal{M}_4$ as well, otherwise apply an independent birth step in~$\mathcal{M}_4$.
\end{compactenum}

If $\RandomYConfiguration_i= \emptyset$ or $\RandomZConfiguration_i= \emptyset$, we perform step~2 only in the process not yet terminated. 
As in~\cref{lem:monotonicity}, the coupling is indeed transparent to each~$\mathcal{M}_j$ and maintains invariants~(i)--(iii).
\end{myproof}

Conclusively, monotonicity and submodularity lead to the following approximation guarantee~\cite{Nemhauser1978}. % which greedy algorithm? defined below.

\begin{restatable}{theorem}{thmpositivebiasapprox}\label{thm:positive_bias_approx}
Given a graph~$G$ and integer~$k$, let~$S^*$  be the seed set that maximizes~$\fp^G_{r}(\SeedSet)$,
and~$S'$ the seed set constructed by a greedy algorithm opting for maximal returns in each step. Then~$\fp^G_{r} (\SeedSet') \geq \left(1-\sfrac{1}{e}\right)\fp^G_{r}(\SeedSet^*)$.
\end{restatable}

%% file: figures/fig_experiments.tex
\begin{figure*}[!t]
\centering
\includegraphics[scale=0.65]{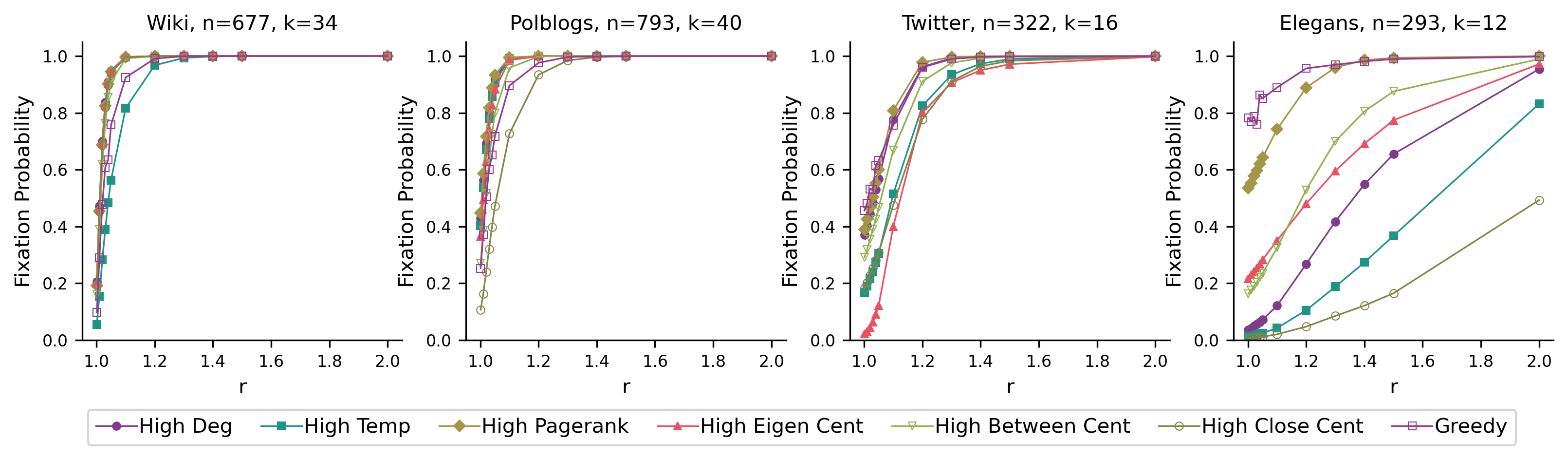}
\caption{
Experimental comparison of heuristics on four graphs.
\label{fig:experiments}
}
\end{figure*}

%% file: sections/05.bias_negative.tex
\section{Invasion with Infavorable Bias}\label{sec:negative_bias}

We now turn our attention to disadvantageous invasions, i.e.,~$r<1$, and show that the problem is intractable. 
To establish our result we focus on the limit~$r\to 0$, i.e., the process is ``infinitely'' biased against invasion,
and write $\fp^G_{0} (\SeedSet) = \lim_{r\to 0}\fp^G_{r}(\SeedSet)$.
%We define~$\fp^G_{0} (\SeedSet) = \lim_{r\to 0}\fp^G_{r}(\SeedSet)$. 
Our key insight is the following lemma.

\begin{restatable}{lemma}{lemnegativebiasvc}\label{lem:negative_bias_vc}
Let $G=(V,E)$ be an undirected graph on $n$ nodes and $S\subseteq V$ a seed set.
If $S$ is a vertex cover on $G$ then $\fp^G_0(S)>1/n^n$, otherwise $\fp^G_0(S)=0$.
\end{restatable}

\begin{myproof}
First, suppose~$\SeedSet$ is a vertex cover.
Denote the nodes in $V\setminus \SeedSet$ by $u_1,\dots,u_i$.
With probability~$1/n^i\ge 1/n^n$,
in the next~$i$ steps the agents at nodes~$u_1,\dots,u_i$
are selected for death (in this order).
Since~$\SeedSet$ is a vertex cover,
each of them adopts the mutant trait, and mutants fixate.

Second, suppose $S$ is not a vertex cover, thus we have an edge $(u,v)\in E$ with $u,v\not \in \SeedSet$.
The idea is that, in the limit $r\to 0$, the agents on $u$ and $v$ are
overwhelmingly more likely to spread to the rest of the graph than to ever become mutants.

\newcommand{\pwin}{p_{\operatorname{ext}}}
\newcommand{\plose}{p_{\operatorname{uv}}}
\newcommand{\pres}{p_{\operatorname{res}}}
\newcommand{\pmut}{p_{\operatorname{mut}}}

Consider a node with $a$ mutant and $b\ge 1$ resident agents among its neighbors.
Note that if an agent at such a node dies,
the node gets occupied by a resident agent with probability
$\frac{b}{b+ar}\ge \frac{1}{n}=:\pres$
and by a mutant agent with probability
$\frac{ar}{b+ar}\le nr=:\pmut$.
Since $G$ is connected, there exists an ordering $(u_1,\dots,u_k)$ of the nodes in $S$
such that node $u_i$ has a neighbor among $\{u_1,\dots,u_{i-1}\}\cup(V\setminus S)$, for each $i=1,\dots,k$.
Consider the first $k$ steps.
With probability $\pwin \ge \left(\frac1n\cdot\pres\right)^k =1/n^{2k}$,
agents at nodes $u_1,\dots,u_k$ die (in this order) and all become residents,
thus mutant extinction occurs.
On the other hand, by Union Bound, within this time-frame, one of $u$, $v$ becomes mutant with probability
$\plose\le \frac2n\pmut\le 2r$.
If neither event occurs, in the next $k$ steps the situation repeats.
Thus, mutants fixate with probability at most $\frac{\plose}{\plose+\pwin}\le
\frac{2r}{2r+1/n^{2k}}
\to_{r\to0} 0$.
\end{myproof}

Since vertex cover is known to be hard even on regular graphs, %~\cite{Feige2003}, 
\cref{lem:negative_bias_vc} implies the following theorem.

\begin{restatable}{theorem}{thmnegativebiasnphard}\label{thm:negative_bias_np_hard}
Fixation maximization in the biased voter process with~$r<1$ is \NP-hard, even on regular undirected graphs.
\end{restatable}

\paragraph{Submodularity considerations.} 
\cref{lem:negative_bias_vc} implies that $\max_{\SeedSet\subseteq V, |\SeedSet|= k} \fp^G_{0}(\SeedSet)$ cannot be approximated within any multiplicative factor in polynomial time,
unless~$\PClass=\NP$, since such an algorithm would decide whether~$\SeedSet$ is a vertex cover on~$G$. 
That is so even though \cref{lem:submodular} implies, by symmetrically exchanging mutant and resident roles, that the \emph{extinction probability} $1-\fp^G_{r}(\SeedSet)$ for~$r<1$ is submodular. 
Fixation maximization with~$r<1$ is thus equivalent to minimizing a submodular function subject to the constraint~$|\SeedSet|=k$. Although the minimization of submodular functions without constraints is polynomially solvable, simple cardinality constraints render it intractable~\cite{Svitkina2011}.
% pk edited 12jan2022

%% file: sections/06.experiments.tex
\section{Experiments}\label{sec:experiments}

Here we present some selective case studies on the performance of various heuristics solving the fixation maximization problem for the biased voter process.
Our studies do not aim to be exhaustive, but rather offer some insights that may drive future work. We consider six heuristics for the choice of the seed set, namely
(i)~high degree,
(ii)~high temperature, defined as $\Temp(u)=\sum_{v} \Weight(v,u)$,
(iii)~high pagerank,
(iv)~high eigenvalue centrality,
(v)~high betweeness centrality,
(vi)~high closeness centrality, 
and the Greedy algorithm (\cref{thm:positive_bias_approx}) that provides an approximation guarantee for~$r\geq 1$. We evaluate all methods on four networks from the Netzschleuder database~\cite{Netzschleuder}, chosen arbitrarily. 
In each case, we choose a budget~$k$ equal to~$5\%$ of the nodes of the graph.

%\Andreas{Graphs to be updated, text as well}

\cref{fig:experiments} shows the performance for various biases~$r\geq 1$. 
The networks Wiki and Polblogs are well-optimizable by almost all heuristics, possibly except high closeness centrality in the second case that performs visibly worse (though not by much).
The Twitter graph is more challenging, while the Elegans graph is the most difficult to optimize, with high variability on heuristic performance.
Unsurprisingly, the Greedy algorithm dominates the performance on this benchmark, 
which is expected given its theoretical guarantees (\cref{thm:positive_bias_approx}).
Among the heuristics, high pagerank appears to be the most consistently performant, and the first one to match the performance of Greedy on the Elegans network.

%% file: sections/07.conclusion.tex
\section{Conclusion}

We have studied the optimization problem of selecting a cardinality-constrained \emph{seed set} of mutants that maximizes the fixation probability in the classic biased voter process. 
We have shown that
%, in contrast to the neutral case~$r=1$ which admits a polynomial-time algorithm~\cite{EvenDar2007}, 
the problem exhibits intricate complexity properties in the presence of biases~$r\neq 1$. 
In particular, the problem is $\NP$-hard in both regimes~$r>1$ and~$r<1$, while the latter case is also hard to approximate. 
On the positive side, we have shown that the optimization function is monotone and submodular, and can thus be approximated efficiently within a factor~$(1-\sfrac{1}{e})$. 
Interestingly, our results on optimization hardness imply that even
just computing the fixation probability over randomly chosen seed sets is also $\NP$-hard.
Random seed sets arise in genetic/biological settings~\cite{Kaveh2015,Allen2020b,Tkadlec2020},
due to the random nature of novel mutations.
Interesting future work includes investigating better approximation factors for $r>1$, and an in-depth search for good heuristics.

%% file: sections/appendix.tex
\section{Appendix}\label{sec:appendix}

\thmubtime*
\begin{proof}{}
It remains to prove that:
\begin{compactenum}
\item When at a bad configuration, the next configuration is good
and the quantity does not change in expectation.
\item When at a good configuration, in one step the quantity increases by at least $s=(r-1)/(r\Delta n)$ in expectation.
\end{compactenum}

Let $X$ be a current configuration.
For any edge $(u,v)$, let
\[
\puv=\frac1n\cdot \frac{ \fit_X(u) }{ \sum_{u'\in N(v)} \fit_X(u')}
\]
be the probability that, in a single step, node $u$ reproduces onto node $v$.
We say that an edge $(u,v)\in E$ is \emph{nice} if $u$ is occupied by a mutant and $v$ by a resident.
Let $Y$ be a (random) configuration after one step from $X$.
Then
\[
\E[\phi(Y)-\phi(X)] = \sum_{(u,v)\text{ is nice}}
%\left(\puv\cdot (+\deg(v)) + \pvu\cdot (-\deg(u))\right).
%\puv\cdot\deg(v) - \pvu\deg(u).
\contrib(u,v),
\]
where $\contrib(u,v)=\puv\cdot\deg(v) - \pvu\deg(u)$ is the contribution of a single nice edge (summed over both directions).

Consider a fixed nice edge $(u,v)$.
To prove the claims in points 1. and 2., it suffices to show that:
\begin{enumerate}
\item[(i)] If all edges incident to $u$ and $v$ connect nodes of different types, then $\contrib(u,v)=0$.
\item[(ii)] Otherwise, $\contrib(u,v)\ge s$, where $s=(r-1)/(r\Delta n)$.
\end{enumerate}

To prove (i), we simply rewrite
\begin{align*}
\contrib(u,v)=\frac1n \frac{r}{r\deg(v)} \deg(v) - \frac1n\frac{1}{\deg(u)}\deg(u)
= 0.    
\end{align*} 
To prove (ii), we distinguish 2 cases:
First, suppose that $v$ has at least one resident neighbor.
Then
\begin{align*}
\contrib(u,v)  &\ge \frac1n\frac{r}{r(\deg(v)-1) +1} \deg(v) -  \frac1n\frac{1}{\deg(u)}\deg(u) \\
&= \frac1n\frac{r-1}{r(\deg(v)-1)+1}
\ge \frac1n \frac{r-1}{r\Delta}=s.
\end{align*}
Similarly, if $u$ has at least one mutant neighbor then
\begin{align*}
\contrib(u,v)  &\ge
\frac1n \frac{r}{r\deg(v)} \deg(v)  - \frac1n\frac{1}{(\deg(u)-1) +r} \deg(u)\\
&= \frac1n\frac{r-1}{\deg(u)-1+r}
\ge \frac1n \frac{r-1}{\Delta}\ge s. \qedhere
\end{align*}
%Each active edge $uv$ with $u$ a mutant and $v$ a resident contributes a non-negative amount to $\E[\phi(\RandomConfiguration_{t+1})-\phi(\RandomConfiguration_t)]$.
%In particular, it contributes 0 iff all edges incident to $u$ and $v$ are active, otherwise it contributes at least $s=(r-1)/(r\Delta n)$ (the ``worst case'' is when all neighbors of $u$ are residents and all but one neighbors of $v$ are mutants).
\end{proof}